\documentclass[conference]{IEEEtran}
\IEEEoverridecommandlockouts
\usepackage[utf8]{inputenc}
\usepackage[utf8]{inputenc} 
\usepackage[T1]{fontenc}
\usepackage{url}
\usepackage{ifthen}
\usepackage{cite}
\usepackage{bbm, dsfont}
\usepackage[cmex10]{amsmath} 

\usepackage{cite}
\usepackage{amsmath,amssymb,amsfonts}
\usepackage{mathtools}
\usepackage{amsthm}
\usepackage{algorithmic}
\usepackage{graphicx}
\usepackage{textcomp}
\usepackage{tikz}
\usepackage{caption}
\usepackage{cuted}
\usepackage{romannum}
\usepackage[utf8]{inputenc}
\usepackage{pgfplots} 
\usepackage{pgfgantt}
\usepackage{pdflscape}
\usepackage{amssymb}
\usepackage{comment}
\usepackage{pst-plot}
\usetikzlibrary{spy}
\usetikzlibrary{positioning,calc}
\usetikzlibrary{decorations.pathmorphing,calc,shapes,shapes.geometric,patterns}
\usetikzlibrary{shapes.multipart}
\usepackage{xfrac}
\usepackage{colortbl}
\usepackage{cancel}
\usetikzlibrary{arrows,positioning,calc,intersections}
\usetikzlibrary{datavisualization.formats.functions}
\def\BibTeX{{\rm B\kern-.05em{\sc i\kern-.025em b}\kern-.08em
    T\kern-.1667em\lower.7ex\hbox{E}\kern-.125emX}}
    
\usepackage{pgfplots}
\usepgfplotslibrary{fillbetween}
\usetikzlibrary{arrows, decorations.markings}
\newtheorem{theorem}{Theorem}
\newtheorem{lemma}{Lemma}

\newtheorem{definition}{Definition}
\newtheorem{claim}{Claim}
\newtheorem{claimproof}{Proof of Claim}

\newcommand{\setd}{\ensuremath{\mathcal{D}}}

\newcommand{\bs}[1]{\boldsymbol{#1}}

\newcommand{\RM}[1]{\MakeLowercase{\romannumeral #1{}}}

\definecolor{calpolypomonagreen}{rgb}{0.12, 0.3, 0.17}
\newcounter{remarkcount}
\newenvironment{remark}{\refstepcounter{remarkcount}\begin{trivlist}\item \textbf{Remark \theremarkcount.}}{\end{trivlist}}
\newcommand{\circlearrow}{}
\DeclareRobustCommand{\circlearrow}{%
  \mathrel{\vphantom{\rightarrow}\mathpalette\circle@arrow\relax}%
}
\newcommand{\circle@arrow}[2]{%
  \m@th
  \ooalign{%
    \hidewidth$#1\circ\mkern1mu$\hidewidth\cr
    $#1-$\cr}%
}
\makeatother
\let\emptyset\varnothing
\usepackage{amsmath, amssymb, amsfonts, amsthm}
\usepackage{bm}
\usepackage{xcolor}
\usepackage{framed}

\usepackage{pgf}
\usepackage{tikz}
\usetikzlibrary{arrows,automata}
\usetikzlibrary{positioning}
\usetikzlibrary{decorations.pathmorphing}
\usetikzlibrary{shapes.geometric}
\usetikzlibrary{fit}
\usetikzlibrary{backgrounds}

\usepackage[caption=false,font=footnotesize]{subfig}

\newcommand{\mbf}{\mathbf}
\newcommand{\mc}{\mathcal}
\newcommand{\mbb}{\mathbb}

\theoremstyle{definition}

\theoremstyle{remark}

\def\BibTeX{{\rm B\kern-.05em{\sc i\kern-.025em b}\kern-.08em
    T\kern-.1667em\lower.7ex\hbox{E}\kern-.125emX}}
\begin{document}
\onecolumn
\thispagestyle{plain}
\pagenumbering{arabic}
\title{A General Formula for Uniform Common Randomness Capacity \\
\thanks{H. Boche was supported in part by the German Federal Ministry of Education and Research (BMBF) within the national initiative on 6G Communication Systems through the research hub 6G-life under Grant 16KISK002, within the national initiative on Post Shannon Communication
(NewCom) under Grant 16KIS1003K. He was further supported by the German Research Foundation (DFG) within Germany’s Excellence Strategy EXC-2092–390781972.
M. Wiese was supported by the Deutsche Forschungsgemeinschaft (DFG, German
Research Foundation) within the Gottfried Wilhelm Leibniz Prize under Grant BO 1734/20-1, and within Germany’s Excellence Strategy EXC-2111—390814868 and EXC-2092 CASA-390781972. C.\ Deppe was supported in part by the German Federal Ministry of Education and Research (BMBF) under Grant 16KIS1005 and in part by the German Federal Ministry of Education and Research (BMBF) within the national initiative on 6G Communication Systems through the research hub 6G-life under Grant 16KISK002.
R. Ezzine was supported by the German Federal Ministry of Education and Research (BMBF) under Grant 16KIS1003K. Thanks goes also to Prof. Watanabe and Prof. Tyagi for discussing some details in this work.}}

\author{
\IEEEauthorblockN{Rami Ezzine\IEEEauthorrefmark{1}, Moritz Wiese\IEEEauthorrefmark{1}\IEEEauthorrefmark{3}, Christian Deppe\IEEEauthorrefmark{2}\IEEEauthorrefmark{4} and Holger Boche\IEEEauthorrefmark{1}\IEEEauthorrefmark{3}\IEEEauthorrefmark{4}}
\IEEEauthorblockA{\IEEEauthorrefmark{1}Technical University of Munich, Chair of Theoretical Information Technology, Munich, Germany\\
\IEEEauthorrefmark{2}Technical University of Munich, Institute for Communications Engineering,  Munich, Germany\\
\IEEEauthorrefmark{3}CASA -- Cyber Security in the Age of Large-Scale Adversaries–
Exzellenzcluster, Ruhr-Universit\"at Bochum, Germany\\
\IEEEauthorrefmark{4}BMBF Research Hub 6G-life, Munich, Germany\\
Email: \{rami.ezzine, wiese, christian.deppe, boche\}@tum.de}
}
\maketitle
\thispagestyle{plain}
\pagenumbering{arabic}
\pagestyle{plain}
\begin{abstract}
We generalize the uniform common randomness capacity formula, initially established by Ahslwede and Csiszár for a two-source model for common randomness generation from independent and identically distributed (i.i.d.) discrete sources with unidirectional communication over rate-limited discrete noiseless channels to the case when the one-way communication is over arbitrary single-user channels. In our proof, we will make use of the transmission capacity formula established by Verdú and Han for arbitrary point-to-point channels.
\end{abstract}


\section{Introduction}
The concept of non-secret common randomness generation was initially introduced by Ahlswede and Csiszár in \cite{part2} for its high relevance in the theory of message identification \cite{identification}, since it may allow a significant increase in the identification capacity of channels\cite{part2,trafo,Ahlswede2021}.

The identification scheme is conceptually different from the classical transmission scheme proposed by Shannon \cite{shannon}. In fact, in the identification framework, the sender is not interested in what the received message is. He rather wants to know if a specific message of special interest to him has been sent or not. Naturally, the sender has no knowledge of that message. Otherwise, the problem would be trivial.
The identification  scheme has striking applications  in digital watermarking \cite{Moulin,watermarkingahlswede, watermarking}. It is also highly relevant in many practical applications which require robust and ultra-reliable low latency information exchange including several machine-to-machine
and human-to-machine systems \cite{applications},  6G communication systems \cite{6Gcomm} and industry 4.0 \cite{industrie40}. It is therefore expected that common randomness will be an important resource for future communication systems \cite{6Gcomm}\cite{6Gpostshannon} and, in particular, that resilience requirements \cite{6Gcomm} and security requirements\cite{semanticsecurity} can also be achieved on the basis of common randomness. These requirements are again of particular importance for achieving trustworthiness, which represents a key challenge for future communication systems due to modern applications\cite{6Gandtrustworthiness}.
For this reason, common randomness generation for future communication networks is an important research question in large 6G research projects \cite{researchgroup1}\cite{researchgroup2}.

Several models for CR generation were introduced in \cite{part2} including the two-source model, which consists of 
two terminals who observe independent and identically distributed (i.i.d.) samples from a known
discrete random source, and who wish to generate a shared random variable with the largest possible probability of agreement by communicating as little as possible.

In \cite{part2}, the authors considered the case when the terminals are allowed to communicate over a classical perfect channel as well as the case when the terminals communicate over a  classical noisy channel. A single-letter expression for the CR capacity 
  for both cases was derived in \cite{part2}. 
  The CR capacity is defined as the maximum amount of shared random variable per channel use. It has been shown in \cite{part2} that the CR capacity can be always  attained with nearly uniform random variables. This is the most convenient form of CR  particularly in the context of key generation \cite{part1}\cite{Maurer}, where further constraints on secrecy are imposed. In our work, however, we will not impose any secrecy requirements.
Later, the results on CR capacity  have been extended in \cite{CRgaussian} to Gaussian channels for their practical relevance in many communication situations such as satellite and deep space
communication links\cite{practicalgaussian}, wired and wireless communications, etc.

  We consider the two-source model for uniform common randomness (UCR) generation in which the two terminals aim to agree on a common uniform or nearly uniform random variable with high probability by communicating as little as possible over an arbitrary single-user channel.
The main contribution of our work consists in establishing a general formula for the UCR capacity  that holds for arbitrary point-to-point channels. In our proof of the UCR capacity formula, we will make use of a well-known result of \cite{verduhan}, which is a general formula  characterizing the transmission capacity of arbitrary channels based on the inf-information rate between the channel inputs and outputs. It is worth-mentioning that no further assumption on stationarity, ergodicity or any kind of  information stability is imposed.  
 For ease of notation, we will assume in the proof of the UCR capacity that the channel input and output alphabets are finite. However, our result does not depend on that assumption. 

\textit{Outline:} The rest of the paper is structured as follows. In Section \ref{sec2}, we present our system model for CR generation, review the definition of an achievable transmission and UCR rate  and present our main result. Section \ref{converseproof} is devoted to the proof of the converse of the UCR capacity, where we will use a change of measure argument introduced in \cite{strongconverse}. In Section \ref{directpart}, we prove the achievability part by extending the UCR generation scheme introduced by Ahlswede and Csiszár to arbitrary single-user channels. Section \ref{conclusion} contains concluding remarks and proposes a potential future work in this field. 

\textit{Notation:} Throughout the paper,  $\log$ is taken to  base 2 and $\ln$ refers to the natural logarithm. For any set $\mc E,$ $\mc E^c$ refers to its complement.  For any random variable $X$ with distribution $P_X,$ $\text{supp}(P_X)$ refers to its support. 
\section{System Model and Definitions}
\label{sec2}
\subsection{System Model}
\label{systemmodel}
Let a discrete memoryless multiple source (DMMS) $P_{XY}$ with two components, with  generic variables $X$ and $Y$ on alphabets $\mathcal{X}$ and $\mathcal{Y}$, respectively, be given. The DMMS emits i.i.d. samples of $(X,Y).$
Suppose that the outputs of $X$ are observed only by Terminal $A$ and those of $Y$ only by Terminal $B.$ We further assume that the joint distribution of $(X,Y)$ is known to both terminals.
Terminal $A$
can communicate with Terminal $B$ over an arbitrary single-user channel $\bs W=\{W_n: \mc T^n \rightarrow \mc Z^n\}_{n=1}^{\infty},$ defined as an arbitrary sequence of $n$-dimensional conditional distributions $W_n$ from $\mc T^n$ to $\mc Z^n$, where $\mc T$ and $\mc Z$ are the  input and output alphabets, respectively.
There are no other resources available to any of the terminals. 
\begin{definition}
A CR-generation protocol of block-length $n$ consists of:
\begin{enumerate}
    \item A function $\Phi$ that maps $X^n$ into a random variable $K$ with alphabet $\mathcal{K}$ satisfying $\lvert \mc K \rvert \geq 3$ generated by Terminal $A.$
    \item A function $\Lambda$ that maps $X^n$ into the channel input sequence $T^n=(T_1,\hdots,T_n)\in \mc T^n.$
    \item A function $\Psi$ that maps $Y^n$ and the channel output sequence $Z^n=(Z_1,\hdots, Z_n)\in \mc Z^n$ into a random variable $L$ with alphabet $\mathcal{K}$ generated by Terminal $B.$
\end{enumerate}
Such a protocol induces a pair of random variables $(K,L)$ whose joint distribution is determined by $P_{XY}$ and by the channel $\bs W$. Such a pair of random variables $(K,L)$ is called permissible.
This is illustrated in Fig. \ref{CRprotocol}.
\end{definition}
\begin{figure}
\centering
\tikzstyle{block} = [draw, rectangle, rounded corners,
minimum height=2em, minimum width=2cm]
\tikzstyle{blockchannel} = [draw, top color=white, bottom color=white!80!gray, rectangle, rounded corners,
minimum height=1cm, minimum width=.3cm]
\tikzstyle{input} = [coordinate]
\usetikzlibrary{arrows}
\begin{tikzpicture}[scale= 1,font=\footnotesize]
\node[blockchannel] (source) {$P_{XY}$};
\node[blockchannel, below=4cm of source](channel) { Channel $\bs W$};
\node[block, below left=2.2cm of source] (x) {Terminal $A$};
\node[block, below right=2cm of source] (y) {Terminal $B$};
\node[above=1cm of x] (k) {$K=\Phi(X^n)$};
\node[above=1cm of y] (l) {$L=\Psi(Y^n,Z^n)$};

\draw[->,thick] (source) -- node[above] {$X^n$} (x);
\draw[->, thick] (source) -- node[above] {$Y^n$} (y);
\draw [->, thick] (x) |- node[below right] {$T^n=\Lambda(X^n)$} (channel);
\draw[<-, thick] (y) |- node[below left] {$Z^n$} (channel);
\draw[->] (x) -- (k);
\draw[->] (y) -- (l);

\end{tikzpicture}
\caption{Two-source model for CR generation with unidirectional communication over an arbitrary point-to-point channel $W.$}
\label{CRprotocol}
\end{figure}
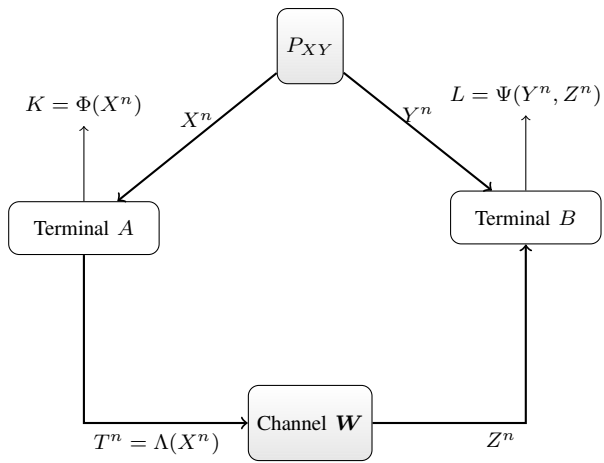
\subsection{Achievable Rate and Capacity}
We define first an achievable UCR rate and the UCR capacity.
\begin{definition} 
\label{ucrrate}
 A number $H$ is called an achievable UCR rate  if there exists a non-negative constant $c$ such that for every $\alpha>0,$ $\beta>0,$ $\delta>0$  and for sufficiently large $n$ there exists a permissible  pair of random variables $(K,L)$ such that
\begin{equation}
   \mbb P\left[K\neq L\right]\leq \alpha, 
    \label{errorcorrelated}
\end{equation}
\begin{equation}
    |\mathcal{K}|\leq 2^{cn},
    \label{cardinalitycorrelated}
\end{equation}
\begin{equation}
\bigg| \frac{1}{n}H(K)-\frac{1}{n}\log \lvert \mc K \rvert \bigg|  \leq \beta,
  \label{uniformity}
\end{equation}
\begin{equation}
    \frac{1}{n}H(K)> H-\delta.
     \label{ratecorrelated}
\end{equation}
\end{definition}
\begin{remark}
Together with \eqref{errorcorrelated}, the technical condition  \eqref{cardinalitycorrelated} ensures for every $\epsilon>0$ and sufficiently large block-length $n$ that 
\begin{align}
    \bigg| \frac{H(K)}{n}-\frac{H(L)}{n} \bigg| \leq \epsilon. \nonumber
\end{align}
\end{remark}
\begin{definition} 
The UCR capacity $C_{UCR}(P_{XY},\bs W)$ is the maximum achievable UCR rate.
\end{definition}
Next, we define an achievable transmission rate and the transmission capacity of the channel $\bs W.$ For this purpose, we begin by providing the definition of a transmission code for the channel $\bs W.$
\begin{definition}
\label{defcode}
A transmission-code $\Gamma_n$ of block-length $n$ and size \footnote{\text{This is the same notation used in} \cite{codingtheorems}.} $\lVert \Gamma_n \rVert $  for the channel $\bs W$ is a family of pairs of codewords and decoding regions $\left\{(\bs{t}_\ell,\setd_\ell) \in \mc T^n \times \mc Z^n: \ell=1,\ldots,\lVert \Gamma_n \rVert \right\}$ such that for all $\ell,j \in \{1,\ldots,\lVert \Gamma_n \rVert\}$  
\begin{align}
&\setd_\ell \cap \setd_j = \emptyset,\quad \ell \neq j. \nonumber
\end{align}The maximum error probability is expressed as 
\begin{align}
    e(\Gamma_n)=\underset{\ell \in \{1,\ldots,\lVert \Gamma_n \rVert\}}{\max}W_{n}({\setd_\ell^{c}}|\mbf{t}_\ell). \nonumber
\end{align}
\end{definition}
\begin{definition}
\label{deftransmissionrate}
   A real number $R$ is called an \textit{achievable} transmission rate of the channel $\bs W$ if for every $\theta,\delta>0$ there exists a code sequence $(\Gamma_n)_{n=1}^\infty$, where each code $\Gamma_n$ of block-length $n$ is defined according to Definition \ref{defcode},  such that
    \[
        \frac{\log\lVert \Gamma_n\rVert}{n}\geq R-\delta
    \]
    and
    \begin{align}
        e(\Gamma_n)\leq\theta
        \nonumber
    \end{align}
    for sufficiently large $n.$ 
\end{definition}
\begin{definition}
The transmission capacity of the channel $\bs W$ is the maximum achievable transmission rate for $W$ and it is denoted by $C(\bs W).$
\end{definition}
\begin{remark}
Throughout the paper, we consider the maximum error probability criterion.
\end{remark}
A general formula for the transmission capacity was established in \cite{verduhan} and it is stated in the following theorem:
\begin{theorem}
\label{generalcapacityformula} \cite{verduhan}
A general formula of the transmission capacity of an arbitrary channel $\bs W$ is equal to 
\begin{align}
C(\bs W)=\underset{\bs{T}}{\sup} \ \underline{I}(\bs{T},\bs{Z}),
\nonumber \end{align}
where $\bs{T}$ is an input process in the form of a sequence of finite-dimensional distributions $\bs{T}=\{T^n=(T_1,\hdots,T_n)\}_{n=1}^{\infty}$
and where $\bs{Z}=\{Z^n=(Z_1,\hdots,Z_n)\}_{n=1}^{\infty}$ is the corresponding output sequence of finite-dimensional distributions induced by $\bs{T}$ via the channel $W.$ Furthermore, $\underline{I}(\bs{T},\bs{Z})$ denotes the inf-information rate between $\bs{T}$ and $\bs{Z},$ which is defined as 
$$ \sup\left\{R: \underset{n\rightarrow\infty}{\lim} \mbb P\left[\frac{1}{n}i(T^n;Z^n)\leq R \right]=0  \right\},  $$
where for any $(t^n,z^n)\in \mc T^n \times \mc Z^n$
\begin{align}
    i(t^n;z^n)=\log\frac{P_{Z^n|T^n}(z^n|t^n)}{P_{Z^n}(z^n)}. \nonumber
\end{align}
\end{theorem}
\subsection{Main Result}
\label{mainresult}
In this section, we give a general formula for the UCR capacity for the model presented in Section \ref{systemmodel}.
\begin{theorem}
For the model in Fig \ref{CRprotocol}, the UCR capacity $C_{UCR}(P_{XY},\bs W)$ is equal to
\begin{equation}
    C_{UCR}(P_{XY},\bs W)= \underset{\substack{U \\{\substack{U \circlearrow{X} \circlearrow{Y}\\ I(U;X)-I(U;Y) \leq C(\bs W)}}}}{\max} I(U;X),
    \nonumber
\end{equation}
with $C(\bs W)$ being the transmission capacity of the channel $\bs W$.
\label{main theorem}
\end{theorem}
The converse proof of Theorem \ref{main theorem} is provided in Section \ref{converseproof} and the direct proof is provided in Section \ref{directpart}.
\section{Converse Proof of Theorem \ref{main theorem}}
\label{converseproof}
Let $(K,L)$ be a permissible pair according to a fixed CR-generation protocol of block-length $n,$ as introduced in Section \ref{systemmodel}. We recall that the latter consists of:
\begin{enumerate}
    \item A function $\Phi$ that maps $X^n$ into a random variable $K$ with alphabet $\mathcal{K}$ satisfying $\lvert \mc K \rvert \geq 3$ generated by Terminal $A.$
    \item A function $\Lambda$ that maps $X^n$ into the channel input sequence $T^n=(T_1,\hdots,T_n)\in \mc T^n.$
    \item A function $\Psi$ that maps $Y^n$ and the channel output sequence $Z^n=(Z_1,\hdots, Z_n) \in \mc Z^n $ into a random variable $L$ with alphabet $\mathcal{K}$ generated by Terminal $B.$
\end{enumerate}
We further assume that $(K,L)$ satisfies $\eqref{errorcorrelated}$  $\eqref{cardinalitycorrelated},$ \eqref{uniformity} and $\eqref{ratecorrelated}.$
We are going to show that there exist arbitrarily large $n$ such that
\begin{align}
    \frac{H(K)}{n} \leq \underset{ \substack{U \\{\substack{U \circlearrow{X} \circlearrow{Y}\\ I(U;X)-I(U;Y) \leq C(\bs W)+\alpha'(n)}}}}{\max} I(U;X), \nonumber 
\end{align}
for some $\alpha'(n)>0$, where $\underset{n\rightarrow \infty}{\lim}\alpha'(n)$ can be made arbitrarily small.
Define $\mu(\beta)=\beta+ 2\beta c+\beta^2,$ where $c>0$ is the constant in \eqref{cardinalitycorrelated} and where $\beta>0$ is the constant in \eqref{uniformity}. 
Let $ \gamma(\alpha,\beta)=2\sqrt{\frac{\sqrt{\mu(\beta)}}{1-\sqrt{\alpha}}},$ with $\alpha>0$ being the constant in \eqref{errorcorrelated}.
Let
\begin{align}
    \kappa(\alpha,\beta)&=\alpha+1-\left(1-4\frac{\mu(\beta)}{\gamma(\alpha,\beta)^2}\right)^2 \nonumber \\
    &=\alpha+1-\left(1-\sqrt{\mu(\beta)}(1-\sqrt{\alpha})\right)^{2}. 
    \nonumber
\end{align}
Assume without loss of generality that $\alpha,\beta>0$ satisfy
\begin{align}
    0<\alpha<1, \label{req1}
\end{align}
\begin{align}
    0<\kappa(\alpha,\beta)<\frac{1}{2}, \label{req2}
\end{align}
and 
\begin{align}
    0<\mu(\beta)<1.\label{req3}
\end{align}
Clearly, \eqref{req1}, \eqref{req2} and \eqref{req3} are satisfied for arbitrarily small positive $\alpha$ and $\beta.$
\begin{lemma}
\label{intervallemma}
For $0<\mu(\beta)<1$ and $0<\alpha<1$
\begin{align}
0<\frac{4\mu(\beta)}{\gamma(\alpha,\beta)^2}<1-\sqrt{\alpha}<1.\nonumber
\end{align}
\end{lemma}
\begin{proof}
We have
\begin{align}
    \gamma(\alpha,\beta)^2=4\frac{\sqrt{\mu(\beta)}}{1-\sqrt{\alpha}} \nonumber
\end{align}
yielding
\begin{align}
    0<4\frac{\mu(\beta)}{\gamma(\alpha,\beta)^2}= \sqrt{\mu(\beta)}(1-\sqrt{\alpha}) <1-\sqrt{\alpha}<1. \nonumber
\end{align}
This proves Lemma \ref{intervallemma}.
\end{proof}

Let $\epsilon$ any positive constant satisfying for infinitely many $n$
\begin{align}
    2\kappa(\alpha,\beta)<\mbb P\left[\frac{1}{n}i(T^n;Z^n) \leq \underline{I}(T;Z)+\epsilon \right]. \label{choiceepsilon} \end{align}
Since
\begin{align}
   \underline{I}(T;Z) = \sup\left\{R: \underset{n\rightarrow\infty}{\lim} \mbb P\left[\frac{1}{n}i(T^n;Z^n)\leq R \right]=0  \right\},\nonumber
\nonumber \end{align}
  and since $0<\kappa(\alpha,\beta)<\frac{1}{2},$  we know that such an $\epsilon>0$ exists.
  
 First, we will introduce and prove the following claim.
 \begin{claim}
 For sufficiently large $n$ satisfying \eqref{choiceepsilon}, it holds that
 \label{claim1}
\begin{align}
    \frac{H(K|Y^n)}{n}\leq C(\bs W)+\alpha'(n), \nonumber
 \end{align}
where $\underset{n\rightarrow \infty}{\lim} \alpha'(n)$ can be made arbitrarily small for $\alpha>0,$ $\beta>0$ and $\epsilon>0$ chosen arbitrarily small.
 \end{claim}
\begin{claimproof}
In order to prove the claim, we will use a change of measure argument. To prepare this, we need some technicalities. 
\begin{lemma}
\label{upperboundvariance}
For $\lvert \mc K \rvert \geq 3,$ it holds for sufficiently large $n$ that
 \begin{align}
\mathrm{var}\left[\frac{1}{n}\log\frac{1}{P_{K}(K)}\right]\leq \mu(\beta).
\nonumber \end{align}
\end{lemma}
\begin{proof}
We have
\begin{align}
    \mbb E \left[\log^{2} P_{K}(K)\right]
    =\frac{1}{\ln(2)^2}\mbb E \left[\ln^{2} P_{K}(K)\right].
\nonumber \end{align}

Define the following two sets
\begin{align}
    \mc K_{L}=\{k\in \mc K: P_{K}(k)\leq \frac{1}{e}  \}
\nonumber \end{align}
and
\begin{align}
    \mc K_{H}=\{k\in \mc K: P_{K}(k)> \frac{1}{e}  \}.
\nonumber \end{align}
Clearly, it holds that $\lvert \mc K_{L} \rvert +\lvert \mc K_{H} \rvert =\lvert \mc K\rvert.$ 
Let 
\begin{align}
    P_{L}=\sum_{k\in \mc K_{L}} P_{K}(k)
\nonumber \end{align}
and
\begin{align}
    P_{H}=\sum_{k\in \mc K_{H}} P_{K}(k).
\nonumber \end{align}

Notice first that
\begin{align}
    1\geq P_{H}> \lvert \mc K_{H} \rvert \frac{1}{e}
\nonumber \end{align}
yielding
\begin{align}
 \lvert \mc K_{H} \rvert<e.
 \nonumber
\end{align}
Therefore,
\begin{align}
    \lvert \mc K_{H} \rvert\leq 2. \nonumber
\end{align}
Since $\lvert \mc K\rvert \geq 3,$ it follows that
\begin{align}
\lvert \mc K _L \rvert=\lvert \mc K \rvert - \lvert \mc K_H\rvert \geq 1. \nonumber
\end{align}

Now, it holds that
\begin{align}
    &\mbb E \left[\ln^{2} P_{K}(K)\right]\nonumber \\
    &=\sum_{k\in \mc K_{L}} P_{K}(k)\ln^2\frac{1}{P_{K}(k)}+ \sum_{k\in \mc K_{H}} P_{K}(k)\ln^2\frac{1}{P_{K}(k)}.
\label{termsinthesum} \end{align}
We we will find appropriate upper-bound for each term in the right-hand side of \eqref{termsinthesum}. On the one hand, we have
\begin{align}
    &\sum_{k\in \mc K_{L}} P_{K}(k)\ln^2\left(\frac{1}{P_{K}(k)}\right) \nonumber \\
    &=P_{L}\sum_{k\in \mc K_{L}} \frac{P_{K}(k)}{P_{L}}\ln^2\left(\frac{1}{P_{K}(k)}\right) \nonumber \\
    &\overset{(a)}{\leq} P_{L}\ln^2\left( \sum_{k\in \mc K_{L}} \frac{P_{K}(k)}{P_{L}}\frac{1}{P_{K}(k)}\right) \nonumber \\
    &=P_{L}\ln^{2}\frac{\lvert \mc K_{L}\rvert}{P_{L}},\nonumber \\
\nonumber \end{align}
where $(a)$ follows because $\ln^{2}(y)$ is concave in the range $y\geq e$ and because for any $k \in \mc K_L,$  $\frac{1}{P_{K}(k)}\geq e.$

On the other hand, we have
\begin{align}
    &\sum_{k\in \mc K_{H}} P_{K}(k)\ln^2\frac{1}{P_{K}(k)} \nonumber\\
    &\overset{(a)}{\leq} \sum_{k\in \mc K_{H}}  P_{K}(k) \ln^2(e) \nonumber \\
    &\leq 1,
\nonumber \end{align}
where $(a)$ follows because $\ln^2(1/y)$ is non-increasing in the range $0<y\leq 1$ and because $\frac{1}{e}<P_{K}(k)\leq 1$ for $k\in \mc K_{H}.$

This implies using the fact that $\lvert \mc K \rvert \geq \lvert \mc K_{L}\rvert\geq 1$
that
\begin{align}
    &\mbb E \left[\ln^{2} P_{K}(K)\right] \nonumber \\
    &\leq 1+ P_{L}\ln^{2}\frac{\lvert \mc K_L\rvert}{P_{L}} \nonumber \\
    &=1+P_{L}\left(\ln\left(\lvert\mc K_L\rvert \right)+\ln\frac{1}{P_{L}}  \right)^{2} \nonumber \\
    &\leq 1+P_{L}\left(\ln\left(\lvert\mc K\rvert\right)+\ln\frac{1}{P_{L}}  \right)^{2} \nonumber \\
    &= 1+P_{L}\left(\ln\left(\lvert\mc K\rvert\right)^2+\ln^{2}\frac{1}{P_{L}}+2\ln\left(\frac{1}{P_{L}}\right) \ln\lvert\mc K\rvert \right) \nonumber \\
    &\overset{(a)}{\leq} 1+\ln\left(\lvert\mc K\rvert\right)^2+\frac{4}{e^2}+2\frac{1}{e}\ln\lvert\mc K\rvert,
\nonumber \end{align}
where $(a)$ follows because $y\ln^2(1/y)$ and $y\ln(1/y)$ are maximized by $\frac{4}{e^2}$ and $\frac{1}{e}$ in the range $0<y\leq 1,$ respectively.

Thus, it follows that 
\begin{align}
    &\mbb E \left[\frac{1}{n^2}\log^{2} P_{K}(K)\right] \nonumber \\
    &\leq\frac{1}{n^2\ln(2)^2}\left(1+\ln\left(\lvert\mc K\rvert\right)^2+\frac{4}{e^2}+2\frac{1}{e}\ln\lvert\mc K\rvert\right) \nonumber \\
    &=\frac{1+\frac{4}{e^2}}{n^2\ln(2)^2}+\frac{\log^{2}\left(\lvert\mc K\rvert\right)}{n^2}+\frac{2\log\lvert \mc K\rvert}{n^2\ln(2)e} \nonumber \\
    &\overset{(a)}{\leq} \frac{1+\frac{4}{e^2}}{n^2\ln(2)^2}+\frac{\log^{2}\left(\lvert\mc K\rvert\right)}{n^2}+\frac{2c}{n\ln(2)e}, \nonumber 
\nonumber \end{align}
where $(a)$ follows because $\frac{\log \lvert \mc K \rvert}{n} \leq c$ 
(from \eqref{cardinalitycorrelated}). 

Since $\underset{n\rightarrow\infty}{\lim}\frac{1+\frac{4}{e^2}}{n^2\ln(2)^2}+\frac{2c}{n\ln(2)e}=0,$ it follows that for sufficiently large $n$
\begin{align}
      &\mbb E \left[\frac{1}{n^2}\log^{2} P_{K}(K)\right] \leq \beta+\frac{\log^{2}\left(\lvert\mc K\rvert\right)}{n^2}. \nonumber
\end{align}

From \eqref{uniformity}, we know that
\begin{align}
  \frac{\log\lvert \mc K \rvert}{n}\leq \frac{H(K)}{n}+\beta.
\nonumber \end{align}
It follows that
\begin{align}
     \mbb E \left[\frac{1}{n^2}\log^{2} P_{K}(K)\right] \leq \beta+\frac{1}{n^2} ( H(K)+n\beta)^2
\nonumber \end{align}
which yields
\begin{align}
    &\mathrm{var}\left[\frac{1}{n}\log\frac{1}{P_{K}(K)}\right] \nonumber \\
    &= \mbb E\left[\frac{1}{n^2}\log^2\left(\frac{1}{P_{K}(K)}\right)\right]-\mbb E\left[ \frac{1}{n} \log\left(\frac{1}{P_{K}(K)}\right) \right]^2 \nonumber \\
    &=\mbb E\left[\frac{1}{n^2}\log^2\left(\frac{1}{P_{K}(K)}\right)\right]-\frac{1}{n^2} H(K)^2 \nonumber \\
    &\leq \beta+ 2\beta \frac{H(K)}{n}+\beta^2 \nonumber \\
    &\overset{(a)}{\leq} \beta+ 2\frac{\beta\log\lvert \mc K \rvert}{n}+\beta^2\nonumber \\
    &\overset{(b)}{\leq} \beta+ 2\beta c+\beta^2 \nonumber \\
    &=\mu(\beta),
\nonumber \end{align}

where $(a)$ follows because $H(K)\leq \log\lvert \mc K \rvert$ and $(b)$ follows from \eqref{cardinalitycorrelated}. 
This proves Lemma \ref{upperboundvariance}.

\end{proof}

Consider now the sets
\begin{align}
    \mc L=\{ k\in \mc K: \frac{1}{n}\log\frac{1}{P_{K}(k)} \geq \frac{1}{n}H(K) -\frac{\gamma(\alpha,\beta)}{2}     \}
\nonumber \end{align}

and
\begin{align}
    \mc D=\{ (k,y^n)\in \mc K\times \mc Y^n:  \frac{1}{n}\log\frac{1}{P_{K|Y^n}(k|y^n)} \geq \frac{1}{n}H(K|Y^n)-\gamma(\alpha,\beta) \}.
\nonumber \end{align}
Let us now introduce and prove the following lemmas:
\begin{lemma}
\label{boundprobsett}
For sufficiently large $n$
\begin{align}
\mbb P\left[K\in \mc L \right]\geq 1-4\frac{\mu(\beta)}{\gamma(\alpha,\beta)^2}. \nonumber
\end{align}
\end{lemma}
\begin{proof}
It holds that
\begin{align}
    &\mbb P\left[ K\notin \mc L\right] \nonumber \\
    &=\mbb P\left[ \frac{1}{n}\log\frac{1}{P_{K}(k)} - \frac{1}{n}H(K) <-\frac{\gamma(\alpha,\beta)}{2}  \right] \nonumber \\
    &\leq \mbb P\left[ \bigg\lvert \frac{1}{n}\log\frac{1}{P_{K}(k)} - \frac{1}{n}H(K) \bigg\rvert > \frac{\gamma(\alpha,\beta)}{2}  \right] \nonumber \\
    &\leq \mbb P\left[ \bigg\lvert \frac{1}{n}\log\frac{1}{P_{K}(k)} - \frac{1}{n}H(K) \bigg\rvert \geq \frac{\gamma(\alpha,\beta)}{2}  \right] \nonumber \\
    &\overset{(a)}{\leq} 4\frac{\mathrm{var}\left[\frac{1}{n}\log\left(\frac{1}{\mbb P_{ K}\left(K\right)}\right)  \right]}{\gamma(\alpha,\beta)^2} \nonumber \\
    &\overset{(b)}{\leq} 4\frac{\mu(\beta)}{\gamma(\alpha,\beta)^2}, \nonumber
\end{align}
where $(a)$ follows from Chebyshev's inequality since $\mbb E\left[\frac{1}{n}\log\frac{1}{P_{K}(k)}\right] =\frac{1}{n}H(K)$ and $(b)$ follows from Lemma \ref{upperboundvariance}.
Therefore
\begin{align}
\mbb P\left[K\in \mc L \right]\geq 1-4\frac{\mu(\beta)}{\gamma(\alpha,\beta)^2}. \nonumber
\end{align}
This proves Lemma \ref{boundprobsett}.
\end{proof}
\begin{lemma}
\label{boundprobsetd}
For sufficiently large $n,$ it holds that
\begin{align}
  \mbb P \left[(K,Y^n)\in \mc D\right]  \geq \left(1-4\frac{\mu(\beta)}{\gamma(\alpha,\beta)^2}\right)^2. \nonumber
\end{align}
\end{lemma}
\begin{proof}
We have
\begin{align}
    &\mbb P \left[(K,Y^n)\in \mc D\right]\nonumber \\ &\geq \sum_{k \in \mc L} \mbb P \left [(K,Y^n)\in \mc D|K=k\right]P_{K}(k) \nonumber \\
    &=\sum_{k \in \mc L}  \mbb P\left[\frac{1}{n}\log\frac{1}{P_{K|Y^n}(k|Y^n)} \geq \frac{1}{n}H(K|Y^n)-\gamma(\alpha,\beta)  \right] P_{K}(k) \nonumber \\
    &=\sum_{k \in \mc L}  \mbb P\left[P_{K|Y^n}(k|Y^n)\leq 2^{n\gamma(\alpha,\beta)-H(K|Y^n)} \right] P_{K}(k) \nonumber \\
     &\geq\sum_{k \in \mc L}  \mbb P\left[P_{K|Y^n}(k|Y^n)< 2^{n\gamma(\alpha,\beta)-H(K|Y^n)} \right] P_{K}(k) \nonumber \\
    &\overset{(a)}{\geq}\sum_{k\in \mc L}\left(1-\frac{P_{K}(k)}{2^{n\gamma(\alpha,\beta)-H(K|Y^n)}}\right) P_{K}(k) \nonumber \\
    &\overset{(b)}{\geq}\left(1-\frac{2^{n\frac{\gamma(\alpha,\beta)}{2}-H(K)}}{2^{n\gamma(\alpha,\beta)-H(K|Y^n)}}\right) \sum_{k\in \mc L} P_{K}(k) \nonumber \\
    &=\left(1-\frac{2^{n\frac{\gamma(\alpha,\beta)}{2}-H(K)}}{2^{n\gamma(\alpha,\beta)-H(K|Y^n)}}\right) \mbb P\left[ K\in \mc L \right] \nonumber \\
    &=\left(1-2^{\left[-n\frac{\gamma(\alpha,\beta)}{2}+H(K|Y^n)-H(K)\right]}\right) \mbb P\left[ K\in \mc L \right] \nonumber \\
    &\overset{(c)}{\geq}(1-2^{-n\frac{\gamma(\alpha,\beta)}{2}})\mbb P\left[ K\in \mc L \right]  \nonumber \\
    &\overset{(d)}{\geq}(1-2^{-n\frac{\gamma(\alpha,\beta)}{2}})(1-\frac{4\mu(\beta)}{\gamma(\alpha,\beta)^2}), \nonumber 
   \end{align}
  where $(a)$ follows from Markov's inequality since $P_{K}(k)=\mbb E\left[P_{K|Y^n}(k|Y^n)\right],$ $(b)$ follows because for $k\in \mc L,$  we know that $ P_{K}(k) \leq 2^{n\frac{\gamma(\alpha,\beta)}{2}-H(K)},$ $(c)$ follows because $H(K|Y^n)-H(K)\leq 0$ and $(d)$ follows from Lemma \ref{boundprobsett}.
  
  Since $\underset{n\rightarrow \infty}{\lim} 1-2^{-n\frac{\gamma(\alpha,\beta)}{2}}=1,$ it follows that for sufficiently large $n$
  \begin{align}
  \mbb P \left[(K,Y^n)\in \mc D\right]  \geq \left(1-4\frac{\mu(\beta)}{\gamma(\alpha,\beta)^2}\right)^2. \nonumber
  \end{align}
  
  This proves Lemma \ref{boundprobsetd}.
\end{proof}

Define now the sets
\begin{align}
     \mc S_{1} =\{ (x^n,y^n,z^n)\in \mc X^{n}\times \mc Y^{n}\times \mc Z^{n}: \frac{1}{n} i(\Lambda(x^n);z^n) \leq \underline{I}(T;Z)+\epsilon\},\nonumber
\end{align}
\begin{align}
    \mc S_{2}=\{(x^n,y^n,z^n) \in \mc X^{n}\times \mc Y^{n}\times \mc Z^{n}: \Phi(x^n)=\Psi(y^n,z^n)   \}, \nonumber
\end{align}
and
\begin{align}
    \mc S_{3}=\{(x^n,y^n,z^n)\in \mc X^{n}\times \mc Y^{n}\times \mc Z^{n}: (\Phi(x^n),y^n)\in \mc D\}. \nonumber
\end{align}

Let $\mc S=\mc S_{1} \cap \mc S_{2} \cap \mc S_{3}.$ Analogously to \cite{strongconverse}, we change the probability measure by defining 
\begin{align}
  &P_{\tilde{X}^n,\tilde{Y}^n,\tilde{Z}^n}(x^n,y^n,z^n) \nonumber \\
  &=\frac{P_{X^n,Y^n,Z^n}(x^n,y^n,z^n) \mbf 1 \left[(x^n,y^n,z^n)\in \mc S  \right]}{\mbb P \left[(X^n,Y^n,Z^n) \in \mc S\right]}, \nonumber
\end{align}
where $\mbf 1[\cdot]$ is the indicator function.
In order to show that $(\tilde{X}^n,\tilde{Y}^n,\tilde{Z}^n)$ is well-defined, we will show that $\mbb P\left[(X^n,Y^n,Z^n)\in \mc S \right]>0$ in what follows.
\begin{lemma} 
For sufficiently large $n$ satisfying \eqref{choiceepsilon}, it holds that
\label{probS}
\begin{align*}
   \mbb P \left[ (X^n,Y^n,Z^n)\in \mc S\right]\geq \kappa(\alpha,\beta)>0. \nonumber  
\end{align*}
\end{lemma}
\begin{proof}
It holds for $   \kappa(\alpha,\beta)=\alpha+\left[1-\left(1-4\frac{\mu(\beta)}{\gamma(\alpha,\beta)^2}\right)^2\right]$ that
\begin{align}
    &\mbb P\left[ (X^n,Y^n,Z^n) \in \mc S\right] \nonumber \\
    &\geq 1- \mbb P\left[(X^n,Y^n) \notin \mc S_{3}\right]-\mbb P\left[(X^n,Y^n,Z^n) \notin \mc S_{2}\right] - \mbb P\left[(X^n,Y^n,Z^n)\notin \mc S_{1}\right]
    \nonumber \\
   &= 1-\mbb P\left[(K,Y^n)\notin\mc D\right]-\mbb P\left[K\neq L\right] -\mbb P\left[(X^n,Y^n,Z^n)\notin \mc S_{1}\right]  \nonumber \\
   &\geq 1-\left[1-\left(1-4\frac{\mu(\beta)}{\gamma(\alpha,\beta)^2}\right)^2\right]-\alpha -\mbb P\left[(X^n,Y^n,Z^n)\notin \mc S_{1}\right] \nonumber \\
   &=1-\kappa(\alpha,\beta)-\mbb P\left[(X^n,Y^n,Z^n)\notin \mc S_{1}\right]  \nonumber \\
   &=1-\kappa(\alpha,\beta)-\mbb P\left[\frac{1}{n}i(T^n;Z^n) > \underline{I}(T;Z)+\epsilon \right] \nonumber \\
   &\overset{(a)}{\geq}  1-\kappa(\alpha,\beta)-\left( 1-2\kappa(\alpha,\beta)\right)                    \nonumber\\
   &=\kappa(\alpha,\beta), \nonumber \\
\nonumber \end{align}
where $(a)$ follows from the choice of $\epsilon$ in \eqref{choiceepsilon}. This proves Lemma \ref{probS}.
\end{proof}
Consider an execution of the protocol of block-length $n$ satisfying \eqref{choiceepsilon} for input $(\tilde{X}^n,\tilde{Y}^n)\sim P_{\tilde{X}^n,\tilde{Y^n}}.$
Let $\tilde{K}=\Phi(\tilde{X}^n).$ Let $\tilde{T}^n=\Lambda(\tilde{X}^n)$ be the new channel input sequence. $\tilde{Z}^n$ is then its corresponding output sequence. We further define $\tilde{L}=\Psi(\tilde{Y}^n,\tilde{Z}^n).$ Here, $\tilde{K}$ is equal to $\tilde{L}$ with probability one. Furthermore, since for every $(x^n,y^n,z^n)\in \mc S,$ 
\begin{align}
      \frac{1}{n}i\left(\Lambda(x^n),z^n\right)\leq \underline{I}(\bs{T};\bs{Z})+\epsilon, \nonumber
\end{align}
we have
\begin{align}
    \mbb E\left[ \frac{1}{n}i\left(\Lambda(\tilde{X}^n),\tilde{Z}^n\right)\right]\leq \underline{I}(\bs{T};\bs{Z})+\epsilon.
\nonumber \end{align}
Therefore
\begin{align}
    \frac{1}{n}I(\tilde{T}^n;\tilde{Z}^n) \leq \underline{I}(\bs{T};\bs{Z})+\epsilon. \label{upperboundnewmutinf} \end{align}
\begin{lemma} It holds that
\label{upperboundlogcardinality}
\begin{align}
    H(K|Y^n) &\leq n\gamma(\alpha,\beta)+\log\frac{1}{\kappa(\alpha,\beta)}+H(\tilde{K}|\tilde{Y}^n). \nonumber
\end{align}
\end{lemma}
\begin{proof}
We have
\begin{align}
   \frac{P_{\tilde{K},\tilde{Y}^n}(k,y^n)}{P_{Y^n}(y^n)}\nonumber 
    &=\frac{1}{P_{Y^n}(y^n)}\sum_{\substack{x^n,z^n \\(x^n,y^n,z^n)\in \mc S \\ \Phi(x^n)=k}} P_{\tilde{X}^n,\tilde{Y}^n,\tilde{Z}^n}(x^n,y^n,z^n)\nonumber \\
    &=\frac{1}{P_{Y^n}(y^n)}\sum_{\substack{x^n,z^n \\(x^n,y^n,z^n)\in \mc S \\ \Phi(x^n)=k}} \frac{P_{X^n,Y^n,Z^n}(x^n,y^n,z^n)}{\mbb P\left[(X^n,Y^n,Z^n)\in \mc S\right]} \nonumber \\
      &\leq \frac{P_{K,Y^n}(k,y^n)}{P_{Y^n}(y^n)\mbb P\left[(X^n,Y^n,Z^n)\in \mc S\right]} \nonumber \\
       &= \frac{P_{K|Y^n}(k|y^n)}{\mbb P\left[(X^n,Y^n,Z^n)\in \mc S\right]}  \nonumber\\
         &\overset{(a)}{\leq} \frac{2^{n\gamma(\alpha,\beta)}}{2^{H(K|Y^n)} \mbb P\left[(X^n,Y^n,Z^n)\in \mc S\right]} \nonumber \\
         &\leq  \frac{2^{n\gamma(\alpha,\beta)}}{2^{H(K|Y^n)} \kappa(\alpha,\beta)} , \nonumber
         \end{align}
where $(a)$ follows because $(k,y^n)\in \mc D.$ 
Therefore for every $(k,y^n) \in \text{supp}(P_{\tilde{K},\tilde{Y}^n}),$ we have
\begin{align}
    2^{H(K|Y^n)} \leq \frac{2^{n\gamma(\alpha,\beta)}}{\kappa(\alpha,\beta)} \frac{1}{\frac{P_{\tilde{K},\tilde{Y}^n}(k,y^n)}{P_{Y^n}(y^n)}}. \nonumber
\end{align}
This implies that for any $(k,y^n) \in \text{supp}(P_{\tilde{K},\tilde{Y}^n}),$ we have
\begin{align}
H(K|Y^n)&\leq \log\frac{2^{n\gamma(\alpha,\beta)}}{\kappa(\alpha,\beta)}-\log\frac{P_{\tilde{K},\tilde{Y}^n}(k,y^n)}{P_{Y^n}(y^n)}. \nonumber 
\end{align}
As a result, it follows that
\begin{align}
 H(K|Y^n) &\leq \log\frac{2^{n\gamma(\alpha,\beta)}}{\kappa(\alpha,\beta)}+\underset{(k,y^n)\in \text{supp}(P_{\tilde{K},\tilde{Y}^n})}{\min}-\log\frac{P_{\tilde{K},\tilde{Y}^n}(k,y^n)}{P_{Y^n}(y^n)} \nonumber   
\end{align}
Now, it holds that
\begin{align}
    &\underset{(k,y^n)\in \text{supp}(P_{\tilde{K},\tilde{Y}^n})}{\min}-\log\frac{P_{\tilde{K},\tilde{Y}^n}(k,y^n)}{P_{Y^n}(y^n)} \nonumber \\ &\leq \mbb E\left[-\log\frac{P_{\tilde{K},\tilde{Y}^n}(\tilde{K},\tilde{Y}^n)}{P_{Y^n}(\tilde{Y}^n)}\right] \nonumber \\ 
    &=\mbb E\left[-\log P_{\tilde{K}|\tilde{Y}^n}(\tilde{K}|\tilde{Y}^n) \right]-\mbb E\left[\log\frac{P_{\tilde{Y}^n}(\tilde{Y}^n)}{P_{Y^n}(\tilde{Y}^n)}\right]\nonumber \\
    &=H(\tilde{K}|\tilde{Y}^n)-D(P_{\tilde{Y}^n}||P_{Y^n}) \nonumber \\ \nonumber \\
    &\leq H(\tilde{K}|\tilde{Y}^n), \nonumber
\end{align}
where $D(\cdot||\cdot)$ denotes the relative entropy. It follows that
\begin{align}
H(K|Y^n) &\leq \log\frac{2^{n\gamma(\alpha,\beta)}}{\kappa(\alpha,\beta)}+H(\tilde{K}|\tilde{Y}^n) \nonumber \\
&=n\gamma(\alpha,\beta)+\log\frac{1}{\kappa(\alpha,\beta)}+H(\tilde{K}|\tilde{Y}^n).\nonumber
\end{align}
This completes the proof of the lemma.
\end{proof}
Now, it holds using lemma \ref{upperboundlogcardinality} that
\begin{align}
   \frac{1}{n} H(K|Y^n) &\leq \frac{1}{n}H(\tilde{K}|\tilde{Y}^n)+\gamma(\alpha,\beta)+\frac{1}{n}\log\frac{1}{\kappa(\alpha,\beta)}.
\label{entropyentropytilde}
\end{align}

We have
\begin{equation}
\frac{1}{n}H(\tilde{K}|\tilde{Y}^{n})=\frac{1}{n}I(\tilde{K};\tilde{Z}^{n}|\tilde{Y}^{n})+\frac{1}{n}H(\tilde{K}|\tilde{Y}^{n},\tilde{Z}^{n}). \label{entropyKtilconditionedYtilde}
\end{equation}
On the one hand, it holds that
\begin{align} 
\frac{1}{n} I(\tilde{K};\tilde{Z}^{n}|\tilde{Y}^{n})&\leq \frac{1}{n} I(\tilde{X}^{n}\tilde{K};\tilde{Z}^{n}|\tilde{Y}^{n}) \nonumber\\
& \overset{(a)}{\leq }\frac{1}{n} I(\tilde{T}^n;\tilde{Z}^n|\tilde{Y}^{n})  \nonumber \\
& =  \frac{1}{n} H(\tilde{Z}^n|\tilde{Y}^{n})-\frac{1}{n}  H(\tilde{Z}^n|\tilde{T}^n,\tilde{Y}^{n}) \nonumber \\
& \overset{(b)}{=} \frac{1}{n}  H(\tilde{Z}^n|\tilde{Y}^{n})-\frac{1}{n}  H(\tilde{Z}^n|\tilde{T}^n) \nonumber \\
& \overset{(c)}{\leq }   \frac{1}{n} H(\tilde{Z}^n)-\frac{1}{n}  H(\tilde{Z}^n|\tilde{T}^n) \nonumber \\
& = \frac{1}{n} I(\tilde{T}^n;\tilde{Z}^n)  \nonumber \\
&\overset{(d)}{\leq} \underline{I}(\bs{T};\bs{Z})+\epsilon,
\nonumber
\end{align}
where $(a)$ follows from the Data Processing Inequality because $\tilde{Y}^{n}\circlearrow{\tilde{X}^{n}\tilde{K}}\circlearrow{\tilde{T}^n}\circlearrow{\tilde{Z}^{n}}$ forms a Markov chain, $(b)$ follows because $\tilde{Y}^{n}\circlearrow{\tilde{X}^{n}\tilde{K}}\circlearrow{\tilde{T}^n}\circlearrow{\tilde{Z}^{n}}$ forms a Markov chain, $(c)$ follows because conditioning does not increase entropy and $(d)$ follows from \eqref{upperboundnewmutinf}.

On the other hand, since $\tilde{K}$ is equal to $\tilde{L}=\Psi(\tilde{Y}^n,\tilde{Z}^n)$ with probability one, it holds that
$$H(\tilde{K}|\tilde{Y}^n,\tilde{Z}^n)=0.$$ 

As a result, we have using  \eqref{entropyKtilconditionedYtilde}
\begin{align}
    \frac{1}{n}H(\tilde{K}|\tilde{Y}^n)&\leq \underline{I}(\bs{T};\bs{Z})+\epsilon. \nonumber 
\end{align}

From \eqref{entropyentropytilde}, it follows that for $\alpha,\beta>0$ satisfying $0<\alpha<1,$ $0<\mu(\beta)<1$ and  $0<\kappa(\alpha,\beta)<\frac{1}{2}$ and for sufficiently large $n$ satisfying \eqref{choiceepsilon}, we have
\begin{align}
    \frac{1}{n}H(K|Y^n)\leq \underline{I}(\bs{T};\bs{Z})+\epsilon+\gamma(\alpha,\beta)+\frac{1}{n}\log\frac{1}{\kappa(\alpha,\beta)}. \nonumber
\end{align}

 From the definition of $\underline{I}(\bs{T};\bs{Z}),$ we know that 
 for any $\zeta>0,$
 \begin{align}
    0< \underset{n\rightarrow\infty}{\lim}\mbb P\left[\frac{1}{n}i(T^n;Z^n) \leq \underline{I}(\bs{T};\bs{Z})+\zeta \right]. \nonumber
 \end{align}
 By choosing $\alpha$ and $\beta$ to be arbitrarily small positive constants, $\mu(\beta)$, $\gamma(\alpha,\beta)$ and $\kappa(\alpha,\beta)$ are made arbitrarily small. We can then choose an arbitrarily small $\epsilon>0$ satisfying 
  \begin{align}
      2\kappa(\alpha,\beta)< \underset{n\rightarrow\infty}{\lim}\mbb P\left[\frac{1}{n}i(T^n;Z^n) \leq \underline{I}(\bs{T};\bs{Z})+\epsilon \right]. \nonumber
 \end{align}

 As a result, for sufficiently large $n$ satisfying $\eqref{choiceepsilon},$ we have
 \begin{align}
     \frac{H(K|Y^n)}{n}\leq \underline{I}(\bs{T};\bs{Z})+\alpha'(n) \nonumber
 \end{align}
for some $\alpha'(n),$ where $\underset{n\rightarrow \infty}{\lim}\alpha'(n)$ can be made arbitrarily small for $\alpha,\beta,\epsilon$ chosen arbitrarily small.
 From Theorem \ref{generalcapacityformula}, we know that
 \begin{align}
     \underline{I}(\bs{T};\bs{Z})\leq C(\bs W).\nonumber
 \end{align}
We deduce then that for sufficiently large $n$ satisfying \eqref{choiceepsilon}
  \begin{align}
     \frac{H(K|Y^n)}{n}\leq C(\bs W)+\alpha'(n). \nonumber
 \end{align}
 This completes the proof of the claim.
\end{claimproof}
 
 Now that we proved the claim, we continue with the converse proof. In our proof, we will use  the following lemma: 
\begin{lemma} (Lemma 17.12 in \cite{codingtheorems})
For arbitrary random variables $S$ and $R$ and sequences of random variables $X^{n}$ and $Y^{n}$, it holds that
\begin{align}
 &I(S;X^{n}|R)-I(S;Y^{n}|R) \nonumber \\ 
 &=\sum_{i=1}^{n} I(S;X_{i}|X_{1},\dots, X_{i-1}, Y_{i+1},\dots, Y_{n},R) \nonumber \\ &\quad -\sum_{i=1}^{n} I(S;Y_{i}|X_{1},\dots, X_{i-1}, Y_{i+1},\dots, Y_{n},R) \nonumber \\
 &=n[I(S;X_{J}|V)-I(S;Y_{J}|V)],\nonumber
\end{align}
where $V=(X_{1},\dots, X_{J-1},Y_{J+1},\dots, Y_{n},R,J)$, with $J$ being a random variable independent of $R$,\ $S$, \ $X^{n}$ \ and $Y^{n}$ and uniformly distributed on $\{1 ,\dots, n \}$.
\label{lemma1}
\end{lemma}
Let $J$ be a random variable uniformly distributed on $\{1,\dots, n\}$ and independent of $K$, $X^n$ and $Y^n$. We further define $U=(K,X_{1},\dots, X_{J-1},Y_{J+1},\dots, Y_{n},J).$ It holds that $U \circlearrow{X_J} \circlearrow{Y_J}.$ \\
Notice  that
{{\begin{align}
H(K)&\overset{(a)}{=}H(K)-H(K|X^{n})\nonumber\\
&=I(K;X^{n}) \nonumber\\
&\overset{(b)}{=}\sum_{i=1}^{n} I(K;X_{i}|X_{1},\dots, X_{i-1}) \nonumber\\
&=n I(K;X_{J}|X_{1},\dots, X_{J-1},J) \nonumber\\
&\overset{(c)}{\leq }n I(U;X_{J}), \nonumber
\end{align}}} where $(a)$ follows because $K=\Phi(X^n)$ and $(b)$ and $(c)$ follow from the chain rule for mutual information.
Applying Lemma \ref{lemma1} for $S=K$, $R=\varnothing$ with $V=(X_1,\hdots, X_{J-1},Y_{J+1},\hdots, Y_{n},J)$ yields
\begin{align}
&I(K;X^{n})-I(K;Y^{n}) 
\nonumber \\&=n[I(K;X_{J}|V)-I(K;Y_{J}|V)] \nonumber\\
&\overset{(a)}{=}n[I(KV;X_{J})-I(K;V)-I(KV;Y_{J})+I(K;V)] \nonumber\\
&\overset{(b)}{=}n[I(U;X_{J})-I(U;Y_{J})], 
\label{UhilfsvariableMIMO1}
\end{align}
where $(a)$ follows from the chain rule for mutual information and $(b)$ follows from $U=(K,V)$. \\
It results using (\ref{UhilfsvariableMIMO1}) that
\begin{align}
n[I(U;X_{J})-I(U;Y_{J})]
&=I(K;X^{n})-I(K;Y^{n}) \nonumber\\
&=H(K)-I(K;Y^{n})\nonumber \\ 
&=H(K|Y^n)
\label{touseclaim}
\end{align}
It follows from \eqref{touseclaim} using Claim \ref{claim1} that for sufficiently large $n$ satisfying \eqref{choiceepsilon}
\begin{align}
    I(U;X_{J})-I(U;Y_{J}) \leq C(\bs W)+\alpha'(n). \nonumber
\end{align}

Since the joint distribution of $X_{J}$ and $Y_{J}$ is equal to $P_{XY}$, for sufficiently large $n$ satisfying $\eqref{choiceepsilon},$ $\frac{H(K)}{n}$ is upper-bounded by $I(U;X)$ subject to $I(U;X)-I(U;Y) \leq C(\bs W) + \alpha'(n)$ with $U$ satisfying $U \circlearrow{X} \circlearrow{Y}$. As a result, for sufficiently large $n$ satisfying \eqref{choiceepsilon}, it follows that
\begin{align}
    \frac{H(K)}{n} \leq \underset{ \substack{U \\{\substack{U \circlearrow{X} \circlearrow{Y}\\ I(U;X)-I(U;Y) \leq C(\bs W)+\alpha'(n)}}}}{\max} I(U;X). \nonumber 
\end{align}
This completes the converse proof of Theorem \ref{main theorem}.
\newpage
\section{Direct Proof of Theorem \ref{main theorem}}
\label{directpart}
The proof of the direct part is rather standard. It suffices to  extend the UCR generation scheme provided in \cite{part2} to arbitrary single-user channels.
\subsubsection{If $ C(\bs W)=0$} In this case, no communication over the channel is possible. Therefore, it follows from \cite{part2} that
    \begin{align}
C_{UCR}(P_{XY},W) &= 
  \underset{ \substack{U \\{\substack{U \circlearrow{X} \circlearrow{Y}\\ I(U;X)-I(U;Y) \leq 0}}}}{\max} I(U;X)  \nonumber \\
  &=  \underset{ \substack{U \\{\substack{U \circlearrow{X} \circlearrow{Y}\\ I(U;X)-I(U;Y) \leq C(\bs W)}}}}{\max} I(U;X).  \nonumber
\end{align}
\subsubsection{If $ C(\bs W)>0$} We extend the coding scheme provided in \cite{part2} to arbitary single-user channels. By continuity, it suffices to show that 
$$ \underset{ \substack{U \\{\substack{U \circlearrow{X} \circlearrow{Y}\\ I(U;X)-I(U;Y) \leq C'}}}}{\max} I(U;X)  $$ is an achievable  UCR rate for every $C'<C(\bs W).$
Let $U$ be a random variable satisfying $U \circlearrow{X} \circlearrow{Y}$ and $I(U;X)-I(U;Y) \leq C'$. Let the maximum error probability $\alpha>0.$ Let $\delta,\beta>0.$ We are going to show that $H=I(U;X)$ is an achievable UCR rate. Without loss of generality, assume that the distribution of $U$ is a possible type for block-length $n$.
For some $\mu>0,$ we let
{{\begin{align}
N_{1}&=\lfloor 2^{n[I(U;X)-I(U;Y)+3\mu]} \rfloor \nonumber\\
N_{2}&=\lfloor 2^{n[I(U;Y)-2\mu]}\rfloor. \nonumber
\end{align}}}For each pair $(i,j)$ with $1\leq i \leq N_{1}$ and $1\leq j \leq N_{2}$, we define a random sequence $\bs{U}_{i,j}\in\mathcal{U}^n$ of type $P_{U}$. Let $\mbf M=\bs{U}_{1,1},\hdots, \bs{U}_{N_{1},N_{2}}$  be the joint random variable of all $\bs{U}_{i,j}s.$ We further define the following two sets which depend on $\mbf M$:
\begin{align}
    S_{1}(\mbf M)&=\{(\bs{x},\bs{y}):(\Phi(\bs{x}),\bs{x},\bs{y}) \in \mathcal{T}_{U,X,Y}^{n}\} \nonumber
\end{align} and
\begin{align}
    S_{2}(\mbf M)=\{(\bs{x},\bs{y}):(\bs{x},\bs{y}) \in S_{1}(\mbf M) \ \text{s.t.} \ \exists \ \bs{U}_{i,\ell}\neq\bs{U}_{i,j}=\Phi(\bs{x}) \nonumber \\   \text{jointly} \ UY\text{-typical with} \ \bs{y} \ (\text{with the same first index} \ i)
\}.\nonumber
\end{align}
It is proved in \cite{part2} that
\begin{align}
    \mathbb{E}_{\mbf M}\left[ \mbb P\left[(X^n,Y^n)\notin  S_{1}(\mbf M)\right]+\mbb P\left[(X^n,Y^n)\in  S_{2}(\mbf M)\right]\right]\leq \zeta(n),
    \label{averagebeta}
\end{align}
where $\zeta(n) \leq \frac{\alpha}{2}$ for sufficiently large $n$. 
We choose a realization $\mbf m=\bs{u}_{1,1},\hdots, \bs{u}_{N_1,N_2}$ satisfying:
\begin{align}
\mbb P\left[(X^n,Y^n)\notin  S_{1}(\mbf m)\right]+\mbb P\left[(X^n,Y^n)\in  S_{2}(\mbf m)\right]\leq \zeta(n).
\nonumber \end{align} 
From \eqref{averagebeta}, we know that such a realization exists.
We assume that each $\bs{u}_{i,j}, i=1\hdots N_1, j=1\hdots N_2,$  is known to both terminals.  
This means that  $N_{1}$ codebooks $C_{i}, 1\leq i \leq N_{1}$, are known to both terminals, where each codebook contains $N_{2}$ sequences, $ \bs{u}_{i,j}, \ j=1,\hdots, N_2$. 

Let $\bs{x}$ be any realization of $X^n$ and $\bs{y}$ be any realization of $Y^n.$
 Let $\Phi(\bs{x})=\bs{u}_{ij}$, if $\bs{u}_{ij}$ is jointly $UX$-typical with $\bs{x}$ (either one if there are several).
 Let $f_1(\bs{x})=i$ if $\Phi(\bs{x}) \in C_{i}$.  If no such a $\bs{u}_{i,j}$ exists, then $f_1(\bs{x})=N_1+1$ and $\Phi(\bs{x})$ is set to a constant sequence $\bs{u}_0$ different from all the ${\bs{u}_{ij}}'s$, jointly $UX$-typical with none of the realizations of $X^n$ and known to both terminals.
  \begin{remark}
 It holds that  $$S_{3}^{c}(\mbf m)\subseteq S_{1}^{c}(\mbf m),$$
 where
 \begin{align}
     S_{3}(\mbf m)=\{\bs{x}: \exists (i,j) \ \text{s.t} \ (\bs{u}_{ij},\bs{x}) \in \mathcal{T}_{U,X}^{n}\}. \nonumber
 \end{align}
 \end{remark}
  Since $ C'<C(\bs W)$, we choose $\mu$ to be sufficiently small such that
      \begin{align}
     \frac{\log \lVert f_1 \rVert}{n}&=\frac{\log(N_1+1)}{n} \nonumber \\
     &\leq C(\bs W)-\mu',
     \label{inequalitylogfSISO}
      \end{align}
for some $\mu'>0,$
 The message $i^\star=f_1(\bs{x})$, with $i^\star\in\{1,\hdots,N_1+1\}$, is encoded to a sequence $\bs{t}$ using a code sequence $(\Gamma^\star_n)_{n=1}^{\infty}$ with rate $\frac{\log \lVert \Gamma^\star_n \rVert}{n}=\frac{\log \lVert f_1 \rVert}{n}$ satisfying \eqref{inequalitylogfSISO}
 and with error probability $e(\Gamma^\star_n)$ satisfying for sufficiently large $n$
 \begin{align}
     e(\Gamma^\star_n) \leq \theta 
 \nonumber \end{align}
where $\theta$ is  a positive constant satisfying $\theta\leq \frac{\alpha}{2}.$ 
  \color{black}
  Here, $\lVert f_1 \rVert$ refers to the cardinality of the set of messages $\{i^\star:i^\star=1,\hdots,N_1+1\}
$.
  From the definition of the transmission capacity, we know that such a code sequence exists. The sequence $\bs{t}$ is sent over the single-user channel. Let $\bs{z}$ be the channel output sequence. Terminal $B$ decodes the message $\tilde{i}^\star$ from the knowledge of $\bs{z}.$
Let $\Psi(\bs{y},\bs{z})=\bs{u}_{\tilde{i}^\star,j}$ if $\bs{u}_{\tilde{i}^\star,j}$ and $\bs{y}$ are jointly $UY$-typical . If there is no such a $\bs{u}_{\tilde{i}^\star,j}$ or there are several, we set $\Psi(\bs{y},\bs{z})=\bs{u}_0$ (since $K$ and $L$ must have the same alphabet).
Now, we are going to show that the requirements in $\eqref{errorcorrelated},$  $\eqref{cardinalitycorrelated},$ \eqref{uniformity} and $\eqref{ratecorrelated}$ are satisfied.
Clearly, (\ref{cardinalitycorrelated}) is satisfied  for $c=I(U;X)+\mu+1$  because
{{\begin{align}
|\mathcal{K}|&=N_1 N_2+1 \nonumber \\
             &\leq  2^{n\left[I(U;X)+\mu\right]}+1 \nonumber \\
             &\leq2^{n\left[I(U;X)+\mu+1\right]}.\nonumber
\end{align}}}We define next for any $(i,j)\in \{1,\hdots,n\}\times\{1,\hdots,n\}$  the set
$$\mc R=\{ \bs{x}\in\mathcal{X}^{n} \ \text{s.t.} \ (\bs{u}_{i,j},\bs{x}) \ \text{jointly} \ UX\text{-typical}\}.$$
Then, it holds that 
\begin{align}
\mbb P[K=\bs{u}_{i,j}] &=\sum_{\bs{x}\in\mc R}\mbb P[K=\bs{u}_{i,j}|X^n=\bs{x}]P_{X}^n(\bs{x}) +\sum_{\bs{x}\in\mc R^c}\mbb P[K=\bs{u}_{i,j}|X^n=\bs{x}]P_{X}^n(\bs{x}) \nonumber \\
&\overset{(\RM{1})}{=}\sum_{\bs{x}\in\mc R}\mbb P[K=\bs{u}_{i,j}|X^n=\bs{x}]P_{X}^n(\bs{x}) \nonumber \\
&\leq \sum_{\bs{x}\in\mc R}P_{X}^n(\bs{x}) \nonumber \\
&=P_{X}^{n}(\{\bs{x}: (\bs{u}_{i,j},\bs{x}) \ \text{jointly} \ UX\text{-typical}\}) \nonumber\\
& = 2^{-nI(U;X)-\kappa(n)}, \nonumber
\end{align}
for some $\kappa(n)>0$ with $\underset{n\rightarrow \infty}{\lim} \frac{\kappa(n)}{n}=0$,
where (\RM{1}) follows because for  $(\bs{u}_{i,j},\mathbf{x})$ being not jointly $UX$-typical, we have $\mbb P[K=\bs{u}_{i,j}|X^n=\bs{x}]=0.$ This yields
{{\begin{align}
H(K)\geq nI(U;X)-\kappa'(n)
\nonumber \end{align}}}
for some $\kappa'(n)>0$ with $\underset{n\rightarrow \infty}{\lim} \frac{\kappa'(n)}{n}=0.$
Therefore, for sufficiently large $n,$ it holds that
\begin{align}
    \frac{H(K)}{n}>H-\delta. \nonumber
\end{align}
Thus, (\ref{ratecorrelated}) is satisfied.
Clearly, it holds that
\begin{align}
    \frac{1}{n}\bigg\lvert H(K)-\log\lvert \mc K \rvert \bigg\rvert \leq \kappa''(n) \nonumber
\end{align}
for some $\kappa''(n)>0$ with $\underset{n\rightarrow\infty}{\lim} \kappa''(n)=0.$ Therefore, for sufficiently large $n,$ $\kappa''(n)\leq \beta.$ This proves \eqref{uniformity}.

 Now, it remains to prove that \eqref{errorcorrelated} is satisfied. For this purpose, we define the following event:
\begin{align}
    \mathcal{D}_{\mbf m}= ``\Phi(X^n) \ \text{is equal to none of the} \  {\bs{u}_{i,j}}'s". \nonumber
\end{align}
We denote its complement by $\mc D_{\mbf m}^{c}.$
We further define $I^\star=f_1(X^n)$ to be the random message generated by Terminal $A$ and  $\tilde{I}^\star$ to be the random message decoded by Terminal $B$. 
We have
\begin{align}
    \mbb P[K\neq L] \nonumber &=\mbb P[K\neq L|I^\star=\tilde{I}^\star]\mbb P[I^\star=\tilde{I}^\star] + \mbb P[K\neq L|I^\star\neq \tilde{I}^\star]\mbb P[I^\star\neq\tilde{I}^\star] \nonumber \\
        &\leq \mbb P[K\neq L|I^\star=\tilde{I}^\star]+ \mbb P[I^\star\neq\tilde{I}^\star].\nonumber
\end{align}
Here,
\begin{align}
    \mbb P[K\neq L|I^\star=\tilde{I}^\star] 
   &= \mbb P[K\neq L|I^\star=\tilde{I}^\star,\mathcal{D}_{\mbf m}]\mbb P[\mathcal{D}_{\mbf m}|I^\star=\tilde{I}^\star] + \mbb P[K\neq L|I^\star=\tilde{I}^\star,\mathcal{D}_{\mbf m}^c]\mbb P[\mathcal{D}_{\mbf m}^c|I^\star=\tilde{I}^\star] \nonumber \\
   &\overset{(\RM{1})}{=}\mbb P[K\neq L|I^\star=\tilde{I}^\star,\mathcal{D}_{\mbf m}^c]\mbb P[\mathcal{D}_{\mbf m}^c|I^\star=\tilde{I}^\star] \nonumber \\
   &\leq \mbb P[K\neq L|I^\star=\tilde{I}^\star,\mathcal{D}_{\mbf m}^c],\nonumber
\end{align}
where $(\RM{1})$ follows from $\mbb P[K\neq L|I^\star=\tilde{I}^\star,\mathcal{D}_{\mbf m}]=0,$ since conditioned on  $I^\star=\tilde{I}^\star$ and $\mathcal{D}_{\mbf m}$, we know that $K$ and $L$ are both equal to $\bs{u}_0$.
It follows that
\begin{align}
    \mbb P[K\neq L] 
    &\leq \mbb P[K\neq L|I^\star=\tilde{I}^\star,\mathcal{D}_{\mbf m}^c]+ \mbb P[I^\star\neq\tilde{I}^\star] \nonumber \\
    &\leq \mbb P\left[(X^n,Y^n)\in  S_{1}^{c}(\mbf m)\cup S_{2}(\mbf m)\right]+\mbb P[I^\star\neq\tilde{I}^\star] \nonumber \\
    &\overset{(a)}{=}\mbb P\left[(X^n,Y^n)\notin  S_{1}(\mbf m)\right]+\mbb P\left[(X^n,Y^n)\in  S_{2}(\mbf m)\right] +\mbb P[I^\star\neq\tilde{I}^\star] \nonumber \\
    &\leq \zeta(n)+ \mbb P[I^\star\neq\tilde{I}^\star],\nonumber\\
    &\leq \zeta(n)+\theta \nonumber \\
    &\leq \alpha.
\nonumber \end{align}
where $(a)$ follows because $S_{1}^{c}(\mbf m)$ and $S_{2}(\mbf m)$ are disjoint. 
This completes the direct proof of Theorem .

\section{Conclusion}
\label{conclusion}
In our work, we studied the problem of UCR generation over arbitrary point-to-point channels. We established a general expression for the UCR capacity that holds for arbitrary single-user channels by making use of the transmission capacity formula elaborated  in \cite{verduhan}.  As a future work, it would be interesting to investigate the problem of CR generation from general sources with one-way communication over arbitrary point-to-point channels.

\vspace{12pt}

\end{document}